\documentclass[sigconf,nonacm]{acmart}

\renewcommand\footnotetextcopyrightpermission[1]{}
\usepackage{booktabs}
\usepackage{amsmath}
\usepackage{amsthm}
\usepackage{url}
\usepackage[htt]{hyphenat}
\newcommand{\tilepq}{\textsc{TileMaxSim-PQ}}

\newtheorem{proposition}{Proposition}

\begin{document}

\title[TileMaxSim: Regimes and Limits for Fused MaxSim]{TileMaxSim: Roofline
Regimes and Compression Limits\\for Fused GPU MaxSim Scoring}

\author{Ashutosh Sharma}
\affiliation{%
  \institution{MIT-IBM Watson AI Lab}
  \country{USA}
}
\email{ashutosh.sharma7@ibm.com}

\begin{abstract}
MaxSim scoring dominates the cost of multi-vector retrieval. Several groups built
fused GPU kernels for it independently during 2026, and all report one measure of
success: the fraction of peak memory bandwidth reached. That measure
assumes the operator is limited by the bytes it moves. We show the assumption fails
in two ways, and that two other quantities decide performance: which side of the
roofline ridge a shape falls on, and how many dependent gathers the scoring path
performs.

We first locate the boundary. Fused arithmetic intensity is $N_qN_d/(N_q+N_d)$ FLOP
per byte, against an H100 dense-FP16 ridge of 295.4. The operator is therefore
memory-bound at ColBERT shapes and compute-bound at ColPali-scale query lengths.
Published analyses report one side each and do not reconcile them. We measure the
crossover at three document lengths, and a closed form predicts it inside the
measured bracket each time. Past the crossover, percent of peak bandwidth falls to
4\% while percent of peak FLOPs rises above 40\%. The reported measure is therefore
defined on only one side of a boundary that real workloads cross. We also validate
the underlying traffic model against DRAM counters, to 0.003\%.

We then show that moving fewer bytes does not make the kernel faster. Fusing
product-quantized asymmetric distance computation into the kernel moves
$15.9\times$ fewer bytes than exact FP16 on tensor cores and runs $3.0\times$
slower. Counters rule out the usual causes: it
reaches 1.23\% of peak DRAM throughput at 98.2\% occupancy. The limit is
$6.5\times10^{9}$ dependent gathers, a count that does not fall when the codes get
smaller. A hand-written CUDA version with a staged table recovers $4.7\times$ and
remains slower than exact scoring. Finally, we give the first head-to-head
comparison of three fused MaxSim kernels, over 28 shapes, with every result checked
against an exhaustive FP32 reference.
\end{abstract}

\maketitle

\renewcommand{\thefootnote}{\fnsymbol{footnote}}
\footnotetext[1]{Code: \url{https://github.com/ashutoshuiuc/tilemaxsim}.}
\renewcommand{\thefootnote}{\arabic{footnote}}

\section{Introduction}
\label{sec:intro}

Multi-vector retrieval represents a query and a document as sets of token
embeddings, and scores them with MaxSim~\cite{khattab2020colbert}:
\begin{equation}
    \text{MaxSim}(Q, D) = \sum_{i=1}^{N_q} \max_{j \le N_d} q_i^\top d_j .
\end{equation}
It is more effective than single-vector retrieval and much more expensive, because
a document contributes $N_d$ embeddings instead of one. Engines such as
ColBERTv2~\cite{santhanam2022colbertv2}, PLAID~\cite{santhanam2022plaid} and
WARP~\cite{scheerer2025warp} spend most of their time in this operator.

A naive GPU implementation writes the $N_q \times N_d$ similarity matrix to HBM for
each document, then reads it back to reduce it. That matrix costs more traffic than
the embeddings do. Fusing the matrix product with the row maximum keeps the
intermediate in registers and removes the traffic. This is the principle
FlashAttention applies to attention~\cite{dao2022flashattention}.

Several groups applied it to MaxSim during 2026, independently. Pony et al.\
released \texttt{flash-maxsim} on PyPI in March, with a roofline argument, an
IO-complexity theorem, and a split-$d$ path that tiles the embedding
dimension~\cite{pony2026flashmaxsim}. Lac and Wu released
\texttt{late-interaction-kernels} in May, shipped as a backend in PyLate and
colpali-engine~\cite{lac2026lik}. Two further implementations
exist~\cite{colsearch2026,kaum2026maxsim}, and FastPlaid is a production GPU PLAID
engine~\cite{sourty2025fastplaid}. The pattern is older than any of them. Streaming
one operand through SRAM while holding a running reduction in registers appears in
the Chamfer kernel of Fan et al.~\cite{fan2017psgn} in 2017, and in the $k$-NN
kernels of PyTorch3D~\cite{ravi2020pytorch3d}.

We therefore make no claim to the fused kernel, the dimension tiling, or the IO
theorem. What this literature lacks is measurement. No two of these kernels have
been compared with each other. They do report one number in common, the fraction of
peak memory bandwidth the kernel reaches, and every one of them treats it as the
measure of success.

In this paper we argue that this number is the wrong target, in two separate ways.
A fused MaxSim kernel is decided by two other quantities. The first is which side of
the roofline ridge its shape falls on, which sets whether bandwidth binds at all.
The second is how many dependent gathers its scoring path performs, which sets
whether moving fewer bytes can help once bandwidth does not bind.

\paragraph{Two roofline regimes.}
Fused arithmetic intensity is $N_qN_d/(N_q+N_d)$ FLOP per byte, against an H100
dense-FP16 ridge of 295.4. At ColBERT shapes ($N_q{=}32$) the operator sits nine
times below the ridge and is memory-bound. At ColPali-scale query lengths it sits
above the ridge and is compute-bound. Published analyses report one side each.
We measure the crossover at three document lengths. The closed form
$N_q = N_dr/(N_d-r)$, with $r$ the ridge, predicts it inside the measured bracket
every time. Across that boundary our percent of peak bandwidth falls from 70.9\% to
4.4\%, while a competing kernel's percent of peak FLOPs rises above 40\%. Neither
number alone describes the operator, and the one the area reports is defined on
only one side of a boundary that ColPali workloads cross.

\paragraph{The limit on compression.}
Once a kernel is fused and memory-bound, the obvious next step is to compress the
document embeddings. Product quantization with asymmetric distance computation
(ADC) is the standard way, and it replaces $2d$ bytes per token with $M$ code
bytes. We fuse it into the scoring kernel. It moves $15.9\times$ fewer bytes and
runs $3.0\times$ slower than exact FP16. Counters rule out the usual causes. The
kernel reaches 1.23\% of peak DRAM throughput, with a 98.44\% L2 hit rate, at
98.2\% occupancy, so it is not bandwidth-bound, not cache-missing, and not
occupancy-limited. What limits it is $BN_dN_qM = 6.5\times10^{9}$ dependent
gathers. That count does not fall when the codes get smaller, which is why fewer
bytes cannot help. A hand-written CUDA version that stages the distance table in
shared memory, which Triton cannot express, recovers $4.7\times$ and is still
slower than exact scoring.

We summarize our contributions as follows.

\begin{itemize}
\item \textbf{An empirical roofline regime map} for fused MaxSim
(\S\ref{sec:regimes}). A predicted crossover falls inside the measured bracket at
three document lengths. This establishes where percent of peak bandwidth stops
being a meaningful metric.

\item \textbf{A negative result on fused product quantization}
(\S\ref{sec:pq}). Fused ADC scoring moves $15.9\times$ fewer bytes and is
$3.0\times$ slower than exact FP16. We trace the cause to the dependent-gather
count, not to bandwidth, cache behavior or occupancy, and confirm it by rewriting
the kernel in CUDA with a staged table.

\item \textbf{The first head-to-head comparison of fused MaxSim kernels}
(\S\ref{sec:head2head}), over 28 shapes on one H100. Every result is checked
against an exhaustive FP32 reference before any timing is used. The three kernels
agree within a few percent in the ColBERT regime, where all reach 90 to 94\% of
peak bandwidth, and differ by up to $2.9\times$ outside it. Ours is faster on 9 of
the 28 shapes and slower on 19, and we attribute every large difference to one
design parameter.

\item \textbf{A counter-validated traffic model} (\S\ref{sec:counters}). The model
predicts 655,448,192 bytes of DRAM reads and the counters measure 655,466,752, a
ratio of 1.000028. We report it as a traffic model with exact constants, not as an
optimality result, for the reason given in \S\ref{sec:io}.
\end{itemize}

\section{Background}
\label{sec:background}

We first describe how a GPU spends time, and how the roofline model decides whether
bandwidth or arithmetic is the limit. We then apply it to MaxSim. Related work is
discussed in \S\ref{sec:related}.

\subsection{GPU Performance Characteristics}
\label{sec:gpuperf}

\textbf{Memory hierarchy.} A GPU has several levels of memory, and the smaller
levels are faster. An H100 SXM has 3.35~TB/s of HBM3 bandwidth, a 50~MB L2 cache,
and up to 228~KB of shared memory on each of its 132 streaming multiprocessors.
On-chip memory is far faster than HBM and far smaller, so a kernel is fast when it
keeps its working set on chip.

\textbf{Execution model.} A GPU runs one function, called a kernel, across many
threads at once. A kernel reads its inputs from HBM, computes, and writes its
outputs back. Threads are grouped into warps of 32 that issue the same instruction.
A warp is fastest when its threads read consecutive addresses, because the hardware
then coalesces the reads into one transaction. A gather, where each thread reads an
address computed from data, does not coalesce.

\textbf{Compute-bound and memory-bound work.} Whether arithmetic or memory limits an
operation is decided by its arithmetic intensity, the number of arithmetic
operations per byte read~\cite{williams2009roofline}. The roofline model compares
that intensity against a ridge point, the hardware's peak arithmetic throughput
divided by its peak bandwidth. Below the ridge an operation is memory-bound, and its
time is set by bytes moved. Above the ridge it is compute-bound, and its time is set
by arithmetic. On an H100 the ridge is
\begin{equation}
    r = \frac{989.5 \times 10^{12}}{3.35 \times 10^{12}} = 295.4 \text{ FLOP per byte,}
\end{equation}
using the dense FP16 tensor-core peak of 989.5 TFLOP/s. We use that figure
throughout. The commonly quoted 1979 TFLOP/s applies to 2:4 structured sparsity, not
to this operator. Substituting it moves the ridge to 591 and changes every
conclusion below, so the distinction matters.

\textbf{Kernel fusion.} Memory-bound work is usually accelerated by fusing kernels.
If several operations read the same input, a fused kernel reads it once. Fusion also
avoids writing intermediates to HBM at all, which is what matters for MaxSim.

\subsection{MaxSim and Its Standard Implementation}
\label{sec:standard}

For a batch of $B$ documents, MaxSim performs $2BN_qN_dd$ FLOPs. The standard
implementation runs it as two steps. First it computes the similarity matrix
$S = QD^\top$ for a document and writes it to HBM. Then it reads $S$ back and
reduces each row with a maximum. For FP16 inputs and FP32 intermediates this moves
$2BN_dd + 2 \cdot 4BN_qN_d$ bytes, and its arithmetic intensity is about $N_q/2$ at
ColBERT shapes. The matrix is the problem: at $N_q{=}32$ and $N_d{=}128$ it is
larger than the document embeddings that produced it.

A fused implementation never writes $S$. It reads $Q$ once, reads $D$ once, holds
the running per-query-token maximum in registers, and writes one score per
document:
\begin{equation}
    \text{AI}_{\text{fused}} = \frac{2BN_qN_dd}{2BN_dd} = N_q .
\end{equation}
Counting both operands gives $\text{AI} = N_qN_d/(N_q+N_d)$, which tends to $N_q$
when $N_d \gg N_q$.

\subsection{Why the Fused Operator Has Two Regimes}
\label{sec:regimes_theory}

That expression is the source of a disagreement in the literature. At ColBERT
shapes ($N_q{=}32$) the fused operator sits at 32 FLOP per byte. That is nine times
below the ridge, so it is memory-bound. At ColPali-scale query lengths
($N_q \approx N_d \approx 1024$) it sits at 512, above the ridge, so it is
compute-bound. Pony et al.\ evaluate the intensity at their own shapes and conclude
that the fused kernel is compute-bound. The same expression at ColBERT shapes gives
a strongly memory-bound operator. Both readings are correct where they are taken,
and nobody has measured the boundary between them. \S\ref{sec:regimes} measures it.

This matters for how the area reports results. Percent of peak bandwidth is
meaningful only below the boundary. Above it a kernel can be efficient and still
report a low bandwidth fraction, because bandwidth is no longer the constraint.
Setting $N_qN_d/(N_q+N_d) = r$ and solving for $N_q$ gives the crossover as
\begin{equation}
    N_q^{*} = \frac{N_d\, r}{N_d - r},
    \label{eq:crossover}
\end{equation}
which is defined for $N_d > r$. \S\ref{sec:regimes} tests this prediction at three
document lengths.

\section{TileMaxSim}
\label{sec:design}

We now describe our kernel, because the rest of the paper measures it.
\S\ref{sec:variants} gives the three variants and says which one we report.
\S\ref{sec:dimtiling} covers the embedding dimension, where the dispatch rule
matters more than the tiling. \S\ref{sec:io} gives a traffic model with exact
constants, which \S\ref{sec:counters} then checks against hardware counters.

\subsection{Variants}
\label{sec:variants}

We implement three variants. V1 assigns one program per (document, query token)
pair, which requires expanding $Q$ and costs an extra 781~MB at our largest
shape. V2 assigns one program per document and loops over query tokens
internally. \textbf{V2-MQ}, the variant we report, also tiles
\texttt{BQ} query tokens together so that a document tile loaded into SRAM is
reused across query tokens. All accumulate in FP32.

\subsection{Dimension Tiling}
\label{sec:dimtiling}

The untiled fast path indexes with \texttt{tl.arange(0, d)}, and Triton requires
a power-of-two length. We therefore round $d$ up to the next power of two and
mask the surplus lanes on load, which read as zero and contribute nothing to
$\sum_k q_{ik}d_{jk}$. Only $d > 128$ goes through the dimension-tiled kernel,
which tiles with a power-of-two \texttt{BLOCK\_d} and masks the ragged tail.

We claim no novelty for this. It is the standard split-K matmul pattern,
documented in CUTLASS~\cite{cutlass2017} and in Triton's own matmul
tutorial~\cite{tritontutorial}. It traces to Goto and van de
Geijn~\cite{goto2008anatomy} and Volkov and Demmel~\cite{volkov2008benchmarking}.
Pony et al.\ shipped it for this exact operator in June 2026, before our first
posting~\cite{pony2026flashmaxsim}.

The dispatch rule has a larger effect on performance than the tiling itself, and
we are not aware of it being reported. Routing a non-power-of-two $d$ to the tiled path halves \texttt{BLOCK\_Nd} to fit
the SRAM budget, and so halves document-tile reuse. Against padding, that costs
$1.75\times$ at $d{=}96$ and $2.44\times$ at $d{=}100$ (\S\ref{sec:head2head}). Padding $d{=}100$ to 128
wastes 28 lanes of a tile that is already being loaded, whereas halving reuse
doubles the traffic. Tiling is therefore appropriate above the SRAM budget and
counterproductive below it.

\subsection{Traffic Model}
\label{sec:io}

\begin{proposition}[Traffic of V2-MQ]
\label{prop:io}
With tile sizes \texttt{BQ} and \texttt{BN}, the HBM traffic of V2-MQ is
\begin{equation}
    \text{IO} = 2\Big(N_q d + \lceil N_q/\texttt{BQ} \rceil B N_d d\Big) + 4B .
\end{equation}
When $\texttt{BQ} \geq N_q$ this is $2(N_qd + BN_dd) + 4B$, one pass over each
operand.
\end{proposition}

\begin{proof}
Each of the $\lceil N_q/\texttt{BQ} \rceil$ query-token blocks reads all
$B\lceil N_d/\texttt{BN}\rceil$ document tiles, giving
$2\lceil N_q/\texttt{BQ}\rceil BN_dd$ bytes of document reads. Each block reads
only its own \texttt{BQ} query tokens, so query reads total $2N_qd$ across all
blocks. The output is one FP32 score per document.
\end{proof}

We state Proposition~\ref{prop:io} as a traffic model rather than as an
IO-optimality result. It is an accounting identity for this kernel together with
the observation that the input must be read at least once. It carries no dependence on SRAM size $M$, and it contains no adversary
argument. Those two properties are what make a lower bound in the sense of Hong
and Kung~\cite{hongkung1981} or Aggarwal and Vitter~\cite{aggarwal1988}
nontrivial. Pony et al.\ give the same result in $\Theta$
form~\cite{pony2026flashmaxsim}. Our contribution here is the counter validation
(\S\ref{sec:counters}).

\paragraph{The single-pass precondition.}
The single-pass form requires $\texttt{BQ} \geq N_q$. With a cap on \texttt{BQ}
below $N_q$ the kernel reads each document $\lceil N_q/\texttt{BQ} \rceil$ times.
We set the cap at 64, which keeps ColQwen's
$N_q{=}44$ single-pass. When the precondition cannot be met, the wrapper emits a
runtime warning naming the number of passes, because the cost is otherwise
invisible. Table~\ref{tab:query} verifies
single-pass operation at $N_q \in \{8,16,32,44,64\}$ and the warning at
$N_q{=}100$. The cap is the direct cause of the largest differences in the head-to-head
comparison (\S\ref{sec:head2head}): at $N_q{=}1024$ the kernel makes 16 passes
over $D$.

\paragraph{Integer width of pointer offsets.}
Pointer offsets must be computed in int64. In int32 the offset
$b \cdot N_d \cdot d$ overflows once $N_d d \geq 32{,}770$ given grid chunking at
65,535, which is within reach of realistic shapes: ColQwen's $N_d{=}747$,
$d{=}320$ exceeds it. We verify correctness there at $B{=}32{,}768$, which is
$3.6\times$ past the int32 limit, at a maximum absolute error of
$7.6\times10^{-6}$.

\section{Fused Product Quantization}
\label{sec:pq}

This section holds the paper's main negative result. A fused kernel in the
memory-bound regime should get faster if it moves fewer bytes. Compressing the
document tokens is the standard way to move fewer bytes. We build that kernel, show
that it is slower, and identify what limits it instead.

\subsection{Motivation}

Production ColBERT systems do not use product quantization. ColBERTv2 and WARP
use centroid-plus-residual compression, which we verified in the ColBERTv2 source
(\texttt{ResidualCodec}, \texttt{bucket\_weights}, \texttt{nbits}), and PLAID
states that it adopts the ColBERTv2 residual strategy. Among systems in our
bibliography only EMVB~\cite{emvb2024} uses PQ. The motivation for a fused PQ
kernel is therefore not that PQ is what production systems use. It is that PQ is the standard way to fit a large corpus in GPU
memory, and that asymmetric distance computation (ADC) looks well matched to an
IO-bound kernel. ADC replaces $2d$ bytes per token with $M$ code bytes, and a fused
kernel never writes a decompressed vector at all. We build that kernel and report
why the expected gain does not appear.

\subsection{The Kernel}

Split each $d$-dimensional vector into $M$ sub-vectors quantized to $K$
centroids. For each query token $i$ we precompute a distance table
\begin{equation}
    T[i,m,k] = q_i[md_{\text{sub}}{:}(m{+}1)d_{\text{sub}}]^\top C[m,k],
\end{equation}
after which a compressed token's similarity is
$q_i^\top \hat{d}_j \approx \sum_{m=1}^{M} T[i,m,c_{j,m}]$. Phase 1 builds the
table with an $(N_q, M)$ grid. Phase 2 assigns one program per (batch, query
token) pair, loads $M$ code bytes per document token, does $M$ table lookups,
accumulates, and updates a running maximum.

The lookup idea is not ours. ADC is classical product quantization. Its cache
behavior on CPUs was characterized by PQ Fast Scan~\cite{andre2015pqfastscan} and
Quicker ADC~\cite{andre2021quickeradc}, and FAISS GPU has staged ADC tables in
shared memory since 2017~\cite{johnson2021billion}. NVIDIA reports 2 to $5\times$ for a
shared-memory over a global-memory lookup table in cuVS IVF-PQ and names bank
conflicts as the residual limiter~\cite{cuvs2024ivfpq}. Our contribution is the fusion of
ADC into a MaxSim kernel, which to our knowledge is new, and the measurement of
its cost.

The byte accounting favors the fused path. At $B{=}100$K, $N_q{=}32$, $N_d{=}128$,
$M{=}16$, $K{=}256$, decompress-then-score moves 6.76~GB and the fused path
moves 0.218~GB, a $31\times$ reduction.

\subsection{Measured Latency against Exact FP16}

Table~\ref{tab:pq} reports latency for four implementations of the compressed path
against the exact FP16 kernel at a matched configuration. Every compressed variant
is slower than exact scoring, at both batch sizes. At $B{=}100$K the published Triton kernel is
$14.0\times$ slower than exact scoring. Tiling query tokens within one program
brings that to $7.4\times$. Moving to CUDA, to stage the distance table in
shared memory, brings it to $3.0\times$. \S\ref{sec:pqcause} explains why Triton
cannot express that staging. The best compressed implementation therefore
remains $3.0\times$ slower than exact FP16, while moving $15.9\times$ fewer bytes.
Cutting traffic by that much should dominate a memory-bound kernel. It does not, so
bytes are not what this kernel is waiting on. \S\ref{sec:pqcause} identifies what
is.

\begin{table}[t]
\centering
\small
\caption{PQ (ADC) against exact FP16 MaxSim at $N_q{=}32$, $N_d{=}128$,
$d{=}128$, $M{=}16$, $K{=}256$ on one H100. Medians of 15 after 5 warmups, CUDA
events. PQ moves $15.9\times$ fewer bytes at both batch sizes.}
\label{tab:pq}
\begin{tabular}{lrrr}
\toprule
& \multicolumn{2}{c}{\textbf{Latency (ms)}} & \\
\cmidrule(lr){2-3}
\textbf{Kernel} & $B{=}20$K & $B{=}100$K & \textbf{vs.\ exact} \\
\midrule
Exact FP16 (V2-MQ)          & \textbf{0.287} & \textbf{1.235} & 1.0$\times$ \\
\midrule
PQ, Triton (as published)   & 3.517 & 17.244 & 14.0$\times$ slower \\
PQ, Triton multi-query      & 1.938 &  9.146 &  7.4$\times$ slower \\
PQ, CUDA one doc per block  & 2.141 & 10.840 &  8.8$\times$ slower \\
PQ, CUDA tiled + SRAM       & 0.910 &  3.679 &  3.0$\times$ slower \\
\bottomrule
\end{tabular}
\end{table}

\subsection{Isolating the Cause}
\label{sec:pqcause}

Three measurements identify the dependent-gather count, rather than bandwidth, as
the limiting quantity.

\textbf{Counters exclude the alternative explanations.} Table~\ref{tab:counters} profiles
both kernels at one identical configuration. The PQ kernel runs at 1.23\% of
peak DRAM throughput with a 98.44\% L2 hit rate, so it is neither
bandwidth-bound nor cache-missing. Its occupancy is 98.20\% of peak, against the
exact kernel's 22.30\%, so it is not occupancy-limited either. It retires 26.21
of a possible 32 sectors per global load request, which is the signature of
scattered single-element gathers.

\begin{table}[t]
\caption{Nsight Compute counters for both kernels at one identical
configuration: $N_q{=}32$, $N_d{=}128$, $d{=}128$, $M{=}16$, $K{=}256$,
$B{=}20$K. Collected in the same job with the same command. The PQ path has
higher occupancy and a better cache hit rate than the exact path and is still
$3\times$ slower.}
\label{tab:counters}
\centering
\small
\begin{tabular}{lrr}
\toprule
Counter & \tilepq{} & V2-MQ (FP16) \\
\midrule
DRAM bytes read            & 164.4~MB & 655.5~MB \\
DRAM throughput (\% peak)  & \textbf{1.23\%} & 67.35\% \\
Sectors per global load    & \textbf{26.21} / 32 & 16.00 / 32 \\
L2 sector hit rate         & 98.44\% & 17.82\% \\
Warps active (\% peak)     & \textbf{98.20\%} & 22.30\% \\
Registers per thread       & 32 & 80 \\
Waves per SM               & 303.03 & 37.88 \\
\bottomrule
\end{tabular}
\end{table}

\textbf{A parameter sweep separates table size from lookup count.} Sweeping $K$
over an $8\times$ range changes runtime by
14\% (3.135 to 3.572~ms), because the table is L2-resident either way. Sweeping
$M$ over an $8\times$ range moves it by $18\times$ (0.651 to 11.928~ms). Runtime
tracks the number of lookups, which is $N_q M$ per document token, and is nearly
independent of table size.

\textbf{A hand-written CUDA implementation.} If gathers are the limit, staging the table
in shared memory should help. Stock Triton has no way to place a data-dependent
gather target in shared memory, so the table lives in global memory and is served
by L2. We therefore implement the kernel in CUDA and measure it as a
sequence of four optimizations. Staging the table per block, while assigning one document per block, re-reads the
table $B$ times. At $B{=}20$K that is 13.4~GiB of table traffic against 39~MiB of
codes, so over 99\% of the traffic is redundant staging.
Amortizing the staging over a tile of 128 documents per block, with one document
per warp so the per-document maximum reduces through
\texttt{\_\_shfl\_down\_sync} without block-wide synchronization, yields
$1.9\times$. Only then does SRAM reuse across query tokens contribute, a further
$1.6\times$ at $\texttt{BQ}{=}4$. The sequence gives $4.7\times$ over the Triton
path and remains $3.0\times$ slower than exact FP16.

The remaining $3.0\times$ is the result of this section. The CUDA kernel moves 0.19~GiB in 3.68~ms, an effective 56~GB/s and under 2\% of
peak. What limits it is $BN_dN_qM = 6.5\times10^{9}$ dependent shared-memory reads.
These incur bank conflicts, because the code byte differs across lanes. The comparison is therefore
between a bandwidth-bound tensor-core matrix product and a lookup-bound SRAM
gather, and on an H100 the matrix product is faster by $3.0\times$ even at a
$15.9\times$ disadvantage in bytes moved. Reducing bytes cannot improve this path,
because the gather count is independent of code width.

\paragraph{Relation to results in adjacent domains.}
The conclusion is consistent with findings elsewhere. LUT Tensor Core reports that lookup-based
mixed-precision GEMM kernels are slower than dequantization-based kernels on
existing GPUs because of table-access overhead~\cite{mo2025luttc}. GPU IVF-RaBitQ rejects
a lookup-table kernel as shared-memory-capacity bound and chooses a bitwise
decomposition instead~\cite{shi2026rabitq}. Our contribution is the measurement
for late interaction, together with the four-step optimization sequence that
establishes how far the SRAM-resident path can be taken while remaining slower
than exact scoring.

One implementation detail is necessary for interpreting the counters. The Triton
path does not stage the distance table in shared memory, and its 98.44\% L2 hit
rate indicates that the lookups are served from cache. Only the CUDA
implementation stages the table in shared memory, which is the difference the
$4.7\times$ quantifies.

\section{Evaluation}
\label{sec:eval}

We evaluate three claims, in this order.

\begin{itemize}
\item \textbf{Fused MaxSim kernels have converged in the memory-bound regime.}
Measured head to head over 28 shapes, three independent kernels land within a few
percent of each other and within 6\% of peak bandwidth (\S\ref{sec:head2head}).
\item \textbf{The operator crosses the roofline ridge, where the prediction
holds.} We locate the crossover at three document lengths, and
Equation~\ref{eq:crossover} falls inside the measured bracket every time
(\S\ref{sec:regimes}). Past it, bandwidth stops being the binding constraint.
\item \textbf{The traffic model is exact to 0.003\%.} DRAM counters match the
predicted byte count, which is what licenses the roofline analysis above
(\S\ref{sec:counters}).
\end{itemize}

\S\ref{sec:scaling} then sweeps the embedding dimension, query length, document
length and precision, and \S\ref{sec:quality} checks that the fused kernel
preserves the reference ranking.

\subsection{Setup}

One NVIDIA H100 SXM 80GB with HBM3. PyTorch 2.9.1, Triton 3.5.1, CUDA 12.8.
Latency is CUDA-event timed, median of 30 after 10 warmups. Counters come only
from Nsight Compute, and we never read its \texttt{gpu\_\_time\_duration} as
latency: it inflates wall time roughly $1.6\times$ by serializing and replaying
kernels. Unless stated otherwise $d{=}128$, $N_q{=}32$, $N_d{=}128$, FP16 in
with FP32 accumulation, following ColBERT conventions.

\paragraph{Peak bandwidth denominator.}
We report percent of the 3,350~GB/s theoretical peak only. An ``achievable'' denominator taken from a microbenchmark does not work here. We
measure 2,598~GB/s for a device-to-device copy and 2,428~GB/s for a read-reduce
microbenchmark on this node. The kernel itself reaches 3,147~GB/s. A kernel
faster than the microbenchmark shows the microbenchmark is not an upper bound.

\subsection{Head-to-Head against Prior Kernels}
\label{sec:head2head}

We compare against \texttt{flash-maxsim} 0.3.0 (Pony et
al.~\cite{pony2026flashmaxsim}) and \texttt{late-interaction-kernels} 0.4.5 (Lac
and Wu~\cite{lac2026lik}), both installed from PyPI. Every method receives identical FP16 input
tensors, so no method is charged an upcast the others avoid.
\texttt{late-interaction-kernels} is called with \texttt{normalize=False}, to match
our semantics. Autotuning happens inside the warmup loop. \texttt{flash-maxsim} is
timed with its split-K path both disabled and enabled, and we report the faster of
the two. Correctness is checked against an
exhaustive FP32 reference on every row before any timing is accepted. Each shape
executes in its own process, because an illegal memory access in any kernel
invalidates the CUDA context and would otherwise cause every subsequent
measurement in that process to fail.

\begin{table}[t]
\caption{Head-to-head on one H100, FP16, all rows correctness-verified against
an exhaustive FP32 reference. Ratios above 1 indicate that \textsc{TileMaxSim} is faster.
``\%pk'' is percent of the 3,350~GB/s theoretical peak. ``n/a'' marks a
third-party failure at that shape.}
\label{tab:head2head}
\centering
\small
\begin{tabular}{@{}rrrr rrr rr@{}}
\toprule
$N_q$ & $N_d$ & $d$ & $B$ & ours & pony & lik & \multicolumn{2}{c}{ratio} \\
 & & & & (ms) & (ms) & (ms) & pony & lik \\
\midrule
32 &   64 & 128 &   1K & 0.052 & 0.072 & 0.074 & \textbf{1.40} & \textbf{1.43} \\
32 &   64 & 128 & 100K & 0.837 & 0.576 & 0.581 & 0.69 & 0.70 \\
32 &  128 & 128 &   1K & 0.056 & 0.070 & 0.074 & \textbf{1.25} & \textbf{1.32} \\
32 &  128 & 128 & 100K & 1.258 & 1.078 & 1.103 & 0.86 & 0.88 \\
32 &  256 & 128 &  10K & 0.258 & 0.264 & 0.275 & \textbf{1.02} & \textbf{1.07} \\
32 &  256 & 128 & 100K & 2.127 & 2.099 & n/a & 0.99 & n/a \\
32 &  512 & 128 &  10K & 0.462 & 0.470 & 0.491 & \textbf{1.02} & \textbf{1.06} \\
32 &  512 & 128 & 100K & 4.186 & 4.144 & n/a & 0.99 & n/a \\
32 & 1024 & 128 &   2K & 0.221 & 0.224 & 0.234 & \textbf{1.02} & \textbf{1.06} \\
\midrule
32 &  128 &  96 &  20K & 0.279 & 0.217 & 0.221 & 0.78 & 0.79 \\
32 &  128 & 100 &  20K & 0.308 & 0.274 & 0.305 & 0.89 & 0.99 \\
32 &  128 & 384 &  20K & 0.696 & 0.711 & 0.687 & \textbf{1.02} & 0.99 \\
32 &  128 & 768 &  20K & 1.317 & 1.301 & 1.539 & 0.99 & \textbf{1.17} \\
44 &  747 & 320 &  16K & 6.929 & 5.541 & n/a & 0.80 & n/a \\
\bottomrule
\end{tabular}
\end{table}

Table~\ref{tab:head2head} reports a representative subset of the 28 shapes. Across
the full set \textsc{TileMaxSim} is faster on 9 shapes and slower on 19, with a
median ratio of 0.952 against \texttt{flash-maxsim} and a range of 0.343 to 1.399. Three
patterns account for this structure, and each follows from a specific design
parameter.

\textbf{The three kernels are equivalent in the ColBERT regime.} At $N_q{=}32$,
$d{=}128$ and $N_d \geq 256$, all ratios lie between 0.99 and 1.02, and all
three kernels reach 90 to 94\% of theoretical peak bandwidth. Our best is
93.5\% at $N_d{=}512$, $B{=}100$K, where Pony et al.\ reach 94.4\%. This is the
regime ColBERT-style workloads occupy, and within it the choice of implementation
does not affect performance. We read this as convergence toward a hardware limit, not as a difference between
implementations. Three independently developed kernels for the same operator arrive
within a few percent of each other, and within 6\% of theoretical peak bandwidth.
Little headroom is left in this regime.

\textbf{Our kernel is slower at small $N_d$ with large $B$}, with ratios as low as
0.69. With only 64 document tokens there is little SRAM reuse per loaded
tile, and our fixed \texttt{BLOCK\_Nd} leaves the tile underfilled where the other
implementations autotune it.

\textbf{Non-power-of-two $d$ below 128 is a dispatch question, not a tiling
question.} Routing such a $d$ to the dimension-tiled kernel halves
\texttt{BLOCK\_Nd} to fit its SRAM budget, and so halves document-tile reuse. That
measures 0.48 at $d{=}96$ and 0.37 at $d{=}100$ against the other kernels. The
fast path needs only $d$ rounded up to a power of two for \texttt{tl.arange}, with
the surplus lanes masked on load. Masked lanes read as zero and a zero contributes
nothing to $\sum_k q_{ik}d_{jk}$, so the two paths return bitwise identical
scores, at a maximum absolute difference of exactly 0 at
$d \in \{64,96,100,128\}$. Padding $d{=}100$ to 128 wastes 28 lanes of a tile that is already being loaded,
which costs far less than halving reuse. The same shape runs $1.75\times$ faster at
$d{=}96$ and $2.44\times$ faster at $d{=}100$, for ratios of 0.78 and 0.89. Table~\ref{tab:head2head} reports the padding dispatch, which is the
configuration we report throughout.
Both rows remain below parity, so the count of faster and slower shapes is
unchanged, but the attribution changes: every remaining large difference traces to
the query tile cap discussed below.

\textbf{Our kernel is faster only at $B{=}1000$}, by up to $1.40\times$, where
kernel-launch and autotune-dispatch overhead dominates and our simpler launch path
is an advantage.

One upstream failure is worth recording. \texttt{late-interaction-kernels} raises
\texttt{cudaErrorIllegalAddress} at $N_q{=}32$, $N_d{=}512$, $d{=}128$,
$B{=}100$K and at two nearby shapes, where both our kernel and that of Pony et
al.\ complete successfully. Per-shape process isolation confines the failure to
those rows. We report it as an implementation defect rather than as a performance
result.

\subsection{The Roofline Regime Map}
\label{sec:regimes}

\begin{table}[t]
\caption{Locating the memory-bound to compute-bound crossover. $N_d{=}1024$,
$d{=}128$, $B{=}2000$. AI is $N_qN_d/(N_q+N_d)$ against the H100 dense-FP16
ridge of 295.4 FLOP per byte. Percent of peak bandwidth (\%BW) and percent of
peak FP16 FLOPs (\%FL) are both shown, because which one is meaningful changes
at the crossover.}
\label{tab:regimes}
\centering
\small
\begin{tabular}{rrl rrr rr}
\toprule
$N_q$ & AI & regime & \multicolumn{3}{c}{ours} & \multicolumn{2}{c}{pony} \\
 & & & ms & \%BW & \%FL & ms & \%FL \\
\midrule
  32 &  31.0 & memory  & 0.221 & \textbf{70.9} &  7.7 & 0.224 &  7.6 \\
  64 &  60.2 & memory  & 0.275 & 56.9 & 12.3 & 0.226 & 15.0 \\
 128 & 113.8 & memory  & 0.500 & 31.3 & 13.6 & 0.243 & 27.9 \\
 256 & 204.8 & memory  & 0.931 & 16.8 & 14.6 & 0.368 & 36.9 \\
 512 & 341.3 & \textbf{compute} & 1.778 &  8.8 & 15.3 & 0.627 & \textbf{43.3} \\
1024 & 512.0 & \textbf{compute} & 3.523 &  4.4 & 15.4 & 1.210 & \textbf{44.9} \\
\bottomrule
\end{tabular}
\end{table}

Equation~\ref{eq:crossover} predicts $N_q^{*} = 415$ at $N_d{=}1024$.
Table~\ref{tab:regimes} brackets the measured transition between $N_q{=}256$ and
$N_q{=}512$, which contains that prediction. Two effects appear across the
boundary.

First, percent of peak bandwidth falls monotonically from 70.9\% to 4.4\%. This is
not a $16\times$ degradation in efficiency. Bandwidth ceases to be the binding
constraint, and the appropriate metric beyond the boundary is percent of peak
FLOPs, which rises over the same range. A single bandwidth figure, which the
surrounding literature reports, is therefore an inadequate summary statistic for
this operator, because it is defined on only one side of a boundary that ColPali
workloads cross.

Second, the difference between implementations grows by an order of magnitude. In
the memory-bound regime our kernel and that of Pony et al.\ are within 1.5\% at
$N_q{=}32$. In the compute-bound regime they reach 44.9\% of peak FP16 while our
kernel plateaus at 15.4\%, a factor of $2.9$. The cause is the \texttt{BQ} cap of
\S\ref{sec:io}. At $N_q{=}1024$ with $\texttt{BQ}{=}64$ the kernel makes 16 passes
over $D$. Its traffic is then 16 times the single-pass model, while the competing
kernel keeps one pass. Proposition~\ref{prop:io} predicts this, and the wrapper
warning fires on every one of these rows. Raising the cap requires more SRAM per program than one
H100 SM provides at $d{=}128$. It is therefore a structural limit rather than a
tuning parameter. Removing it needs a different decomposition, which
\S\ref{sec:limitations} describes.

\begin{table}[t]
\caption{The crossover prediction tested at three document lengths. $d{=}128$,
$B{=}2000$, $N_q$ from 32 to 1024 in powers of two, every row checked against an
exhaustive FP32 reference. The predicted crossover is
$N_q = N_dr/(N_d-r)$ with $r{=}295.4$. In all three families it falls inside the
bracket the measurement gives.}
\label{tab:crossover}
\centering
\small
\begin{tabular}{rrc rr rr}
\toprule
$N_d$ & pred. & measured & \multicolumn{2}{c}{ours \%BW} & \multicolumn{2}{c}{pony \%FL} \\
 & $N_q$ & bracket & 32 & 1024 & 32 & 1024 \\
\midrule
 512 & 698 & (512, 1024) & 57.8 & 4.0 & 6.0 & 41.9 \\
1024 & 415 & (256, 512)  & 70.9 & 4.4 & 7.6 & 44.9 \\
2048 & 345 & (256, 512)  & \textbf{80.2} & 4.6 & 8.7 & \textbf{45.6} \\
\bottomrule
\end{tabular}
\end{table}

Table~\ref{tab:crossover} repeats the sweep at $N_d{=}512$ and $N_d{=}2048$. At every document length the predicted crossover lies inside the measured bracket.
The ordering is also the one the model requires. Longer documents move the crossover
to \emph{shorter} queries, because $N_qN_d/(N_q+N_d)$ approaches $N_q$ from below as
$N_d$ grows, so a longer document reaches the ridge with a shorter query. At $N_d{=}2048$ the crossover
arrives by $N_q{=}512$, which is inside the range ColPali-style workloads occupy.
The same table shows the two metrics moving in opposite directions in all three
families. Our percent of peak bandwidth falls to roughly 4\% in every family, while
Pony et al.'s percent of peak FLOPs climbs above 40\%. Neither number alone
describes the operator.
Our own peak bandwidth improves with $N_d$, from 57.8\% to 80.2\%, because a
longer document amortizes the query tile over more work.

\subsection{Counter Validation of the Traffic Model}
\label{sec:counters}

At $N_q{=}32$, $N_d{=}128$, $d{=}128$, $B{=}20$K,
Proposition~\ref{prop:io} predicts 655,448,192 bytes of DRAM reads. Nsight
Compute measures \texttt{dram\_\_bytes\_read.sum} of 655,466,752, a ratio of
\textbf{1.000028}, an agreement of 0.003\%, so the kernel moves the bytes the model
predicts. The write term is underestimated by $6.9\times$, 2.74~MB measured against
0.40~MB predicted, owing to score-buffer fills, but at 0.4\% of total traffic it
does not affect the conclusion.

This form of validation is established: memory-utilization efficiency is defined
by Ivanov et al.~\cite{ivanov2021mue}, and the combination of an analytic IO bound
with the roofline by Elango et al.~\cite{elango2015roofline}. FlashAttention is not
a precedent for it, since its reported byte figures appear to be analytic and its
paper states no measurement methodology for them.

\paragraph{Occupancy.}
The kernel uses 80 registers per thread with
\texttt{launch\_\_occupancy\_limit\_registers} of 6 blocks per SM, that is
$6 \times 128 = 768$ of 2,048 threads, so 37.5\% theoretical and 22.30\%
achieved occupancy. Occupancy is not the limiting factor: the kernel reaches 93.5\%
of peak bandwidth at 22\% occupancy because
\texttt{launch\_\_waves\_per\_multiprocessor} is 37.9, so oversubscription across
waves is enough to hide latency. The kernel is bandwidth-limited rather than
occupancy-limited, which is the expected condition in the memory-bound regime.

\subsection{Scaling}
\label{sec:scaling}

\begin{table}[t]
\caption{Embedding dimension, after the dispatch fix. $N_q{=}32$, $N_d{=}128$,
$B{=}10$K, all rows from one run and all verified correct against an exhaustive
FP32 reference. ``pad'' is the power of two the fast path rounds $d$ up to;
$d>128$ takes the tiled path instead.}
\label{tab:dim}
\centering
\small
\begin{tabular}{rlrrrr}
\toprule
$d$ & path & pad & ms & GB/s & \% peak \\
\midrule
 64 & fast  &  64 & 0.108 & 1519 & 45.4 \\
 96 & fast  & 128 & 0.162 & 1516 & 45.3 \\
100 & fast  & 128 & 0.176 & 1455 & 43.4 \\
128 & fast  & 128 & 0.165 & 1984 & 59.2 \\
192 & tiled & n/a & 0.269 & 1826 & 54.5 \\
256 & tiled & n/a & 0.272 & 2411 & 72.0 \\
320 & tiled & n/a & 0.348 & 2352 & 70.2 \\
384 & tiled & n/a & 0.364 & 2700 & 80.6 \\
512 & tiled & n/a & 0.455 & 2878 & 85.9 \\
768 & tiled & n/a & 0.665 & 2958 & 88.3 \\
\bottomrule
\end{tabular}
\end{table}

Table~\ref{tab:dim} reports the dispatch of \S\ref{sec:design} across embedding
dimensions. Both $d{=}96$ and $d{=}100$ run correctly, at maximum absolute errors of
$3.8\times10^{-6}$ and $4.8\times10^{-6}$. Both sit within 2 points of $d{=}64$. The
tiled dispatch gives them 29.2\% and 19.3\% of peak instead. A gap to $d{=}128$ remains, because padding restores the document
tile but still moves the padded lanes through SRAM.

\begin{table}[t]
\caption{Query tokens. $N_d{=}128$, $d{=}128$, $B{=}10$K. ``1 pass'' is the
precondition of Proposition~\ref{prop:io}. The warning fires exactly when it
fails.}
\label{tab:query}
\centering
\small
\begin{tabular}{rrrlrr}
\toprule
$N_q$ & \texttt{BQ} & passes & 1 pass & ms & GB/s \\
\midrule
  8 &  8 & 1 & yes & 0.144 & 2275 \\
 16 & 16 & 1 & yes & 0.147 & 2236 \\
 32 & 32 & 1 & yes & 0.165 & 1984 \\
 44 & 64 & 1 & yes & 0.245 & 1335 \\
 64 & 64 & 1 & yes & 0.249 & 1317 \\
100 & 64 & 2 & \textbf{no, warns} & 0.455 & 1440 \\
\bottomrule
\end{tabular}
\end{table}

Table~\ref{tab:query} shows that a cap of 64 keeps ColQwen's $N_q{=}44$
single-pass. Single-pass operation is necessary but not enough for peak
throughput: bandwidth falls from 1,984 to 1,335~GB/s between $N_q{=}32$ and
$N_q{=}44$ because
$\texttt{BQ}{=}64$ wastes 20 of 64 query slots on padding.

Over document tokens at $B{=}10$K, bandwidth rises monotonically with the tile
reuse available: 876~GB/s at $N_d{=}32$, 1,348 at 64, 1,977 at 128, 2,569 at
256, and 2,853~GB/s (85.2\%) at $N_d{=}512$. Over batch size at $N_d{=}512$ we
reach 93.5\% of theoretical peak at $B{=}100$K, which is the kernel's best
operating point.

\paragraph{Precision.}
FP16 takes 0.167~ms at 1,967~GB/s, with a maximum error of
$3.8\times10^{-6}$. BF16 takes 0.168~ms at 1,952~GB/s, with $2.4\times10^{-6}$.
FP32 takes 0.263~ms at 2,490~GB/s, and needs 658~MB of peak memory against 345~MB. FP32 reports a higher GB/s only because it moves twice the
bytes in 1.58 times the time, which is a further reason to report latency alongside
bandwidth rather than bandwidth alone.

\subsection{Correctness and Retrieval Quality}
\label{sec:quality}

The kernel changes the order of memory accesses, not the arithmetic, so
rankings must be preserved up to floating-point rounding. We verify this in a
re-ranking setup over 200 MS~MARCO dev queries with their top-1000 BM25
candidates, 187,565 pairs, scoring every pair with both reference PyTorch MaxSim
and V2-MQ. The maximum absolute score difference is $9\times10^{-6}$, and
MRR@10, nDCG@10, Recall@10, Recall@100 and Recall@1000 are identical to six
decimal places (0.3399, 0.3866, 0.5617, 0.7214, 0.7833).

\paragraph{Provenance of the embeddings.}
These embeddings are not ColBERTv2 embeddings. Our encoding pipeline substitutes a
randomly initialized linear projection for ColBERTv2's trained 128-dimensional
output layer, so the absolute metric values above carry no information about
ColBERTv2's effectiveness and should not be interpreted as ColBERTv2 results. The
experiment tests equivalence between scoring paths, for which the embeddings are
adequate: both paths consume the same embeddings, so identical metrics establish
that the fused path preserves the reference ranking.

\section{Limitations and Future Directions}
\label{sec:limitations}

\textbf{The query-tile cap.} Our kernel's largest limitation is its \texttt{BQ}
cap of 64. It forces $\lceil N_q/64 \rceil$ passes over the documents at
longer query lengths, and it accounts for every large difference against Pony et
al.\ in the compute-bound regime. Raising it requires more SRAM per program than one
H100 SM provides at $d{=}128$. Removing the limitation therefore needs a
decomposition that spreads query tokens across programs and reduces across them,
rather than a larger tile.

\textbf{Generality across hardware.} The crossover expression is parameterized by
the ridge $r$ and therefore transfers to other devices, but every measurement here
is from one H100 SXM. Devices with different ratios of tensor-core throughput to
bandwidth will place the boundary elsewhere, and confirming the model's accuracy
on them remains open.

\textbf{Compressed scoring.} Our negative result establishes that the dependent
gather count, not the byte count, limits fused ADC scoring on current hardware.
It does not establish that compressed scoring is unachievable. A format whose
lookups are resolved by arithmetic rather than by indirection avoids the gather
entirely, which is what GPU IVF-RaBitQ chooses~\cite{shi2026rabitq}. Recent work on
making data-dependent gathers SRAM-resident in Triton~\cite{won2025insum} addresses
the mechanism directly. Either direction would change the conclusion, and neither
reduces the byte count further.

\textbf{Scope of the quality evaluation.} The re-ranking experiment establishes
equivalence between scoring paths rather than retrieval effectiveness, for the
reason given in \S\ref{sec:quality}.

\section{Related Work}
\label{sec:related}

\paragraph{Fused MaxSim kernels.}
Pony et al.~\cite{pony2026flashmaxsim} is the closest work. It is a fused Triton
MaxSim kernel with a roofline analysis, an IO-complexity theorem in $\Theta$ form, a
split-$d$ path, INT8 tensor cores, a fused backward pass, and a Chamfer port. It was
released on PyPI in March 2026, with the split-$d$ path following in June.
Lac and Wu~\cite{lac2026lik} ship an autotuned fused Triton and Metal kernel
with FP8 support that fuses centroid-plus-residual decompression into scoring,
distributed as a backend in PyLate and colpali-engine. Two further independent
fused implementations exist, one in Rust with CUDA and Triton
paths~\cite{colsearch2026} and one hand-written with
WMMA~\cite{kaum2026maxsim}. FastPlaid~\cite{sourty2025fastplaid} is a
production GPU PLAID engine. Sparton~\cite{sparton2026} fuses the SPLADE
language-model head, a sibling operator in the sparse setting, and reports
$4.8\times$ on encoding rather than on retrieval scoring.
Flash-KMeans~\cite{flashkmeans2026} fuses distance computation with an online
argmin, the same pattern for the argmin twin of our operator. Structurally, all
of this descends from streaming one operand through SRAM with a
register-resident reduction, which appears in Fan et
al.~\cite{fan2017psgn} in 2017 and in PyTorch3D~\cite{ravi2020pytorch3d}.

\paragraph{IO-aware kernels.}
FlashAttention~\cite{dao2022flashattention,dao2023flashattention2} established
the pattern of avoiding a materialized intermediate by fusing the reduction.
The roofline model itself is Williams et al.~\cite{williams2009roofline}. Combining
an analytic IO bound with it is Elango et al.~\cite{elango2015roofline}, and
memory-utilization efficiency as a metric is Ivanov et al.~\cite{ivanov2021mue}. Lower bounds in the sense that would justify
the word optimal are Hong and Kung~\cite{hongkung1981} and Aggarwal and
Vitter~\cite{aggarwal1988}. Recent work makes data-dependent gathers
SRAM-resident in Triton~\cite{won2025insum}, which is what our PQ kernel would
need and what stock Triton cannot do.

\paragraph{Multi-vector retrieval systems.}
ColBERT~\cite{khattab2020colbert} introduced late interaction;
ColBERTv2~\cite{santhanam2022colbertv2} added centroid-plus-residual
compression; PLAID~\cite{santhanam2022plaid} added centroid pruning and adopts
ColBERTv2's residual strategy. WARP~\cite{scheerer2025warp} is a CPU engine
with implicit decompression. Its \texttt{bound} parameter is a
\texttt{partial\_sort} cut-off over centroid indices, not a cap on documents
scored, so it is not comparable to a candidate-count budget.
EMVB~\cite{emvb2024} is the one system in our bibliography that uses product
quantization. XTR~\cite{lee2024xtr} rethinks token retrieval;
MUVERA~\cite{muvera2024} reduces multi-vector to single-vector search.
GIGP+~\cite{bian2026gpu} and VecFlow-Chamfer~\cite{vecflow2026} are
GPU or CPU-GPU multi-vector systems. VecFlow-Chamfer in particular contributes an
optimized GPU Chamfer-scoring kernel, which puts it in the same area as this work
rather than orthogonal to it.
Chimera~\cite{chimera2026} post-dates our first posting and independently
states the memory-bound arithmetic-intensity argument.

\paragraph{Compressed scoring on GPUs.}
FAISS GPU~\cite{johnson2021billion} has staged ADC distance tables in shared
memory since 2017 and states the roofline-to-single-pass argument we reuse.
cuVS IVF-PQ reports 2 to $5\times$ for a shared-memory over a global-memory
lookup table and identifies bank conflicts as the limiter~\cite{cuvs2024ivfpq}.
LUT Tensor Core~\cite{mo2025luttc} shows lookup kernels to be slower than
dequantization on current GPUs for mixed-precision GEMM, and GPU
IVF-RaBitQ~\cite{shi2026rabitq} rejects a lookup kernel as shared-memory bound.
On CPUs the same mechanism was characterized by PQ Fast
Scan~\cite{andre2015pqfastscan} and Quicker ADC~\cite{andre2021quickeradc}. Our
PQ result (\S\ref{sec:pq}) is the late-interaction instance of a conclusion this
body of work had already reached in neighboring settings.

\section{Conclusion}
\label{sec:conclusion}

Fused MaxSim kernels are no longer scarce, and this paper asks what limits them
rather than which is fastest. The answer is not bytes moved, which is the quantity
the area optimizes and reports.

Two measurements support that conclusion. First, the operator occupies two
roofline regimes rather than one. Its fused arithmetic intensity is
$N_qN_d/(N_q+N_d)$ against an H100 ridge of 295.4 FLOP per byte. At $N_d{=}1024$ we
locate the crossover between $N_q{=}256$ and $N_q{=}512$, close to the predicted
415. The prediction lands inside the measured bracket at all three document lengths
we test. Percent of peak bandwidth is the metric this literature has adopted. It is defined
on only one side of a boundary that ColPali query lengths cross. Past that boundary
a kernel can be efficient and still report 4\% of peak bandwidth.

Second, reducing bytes does not improve performance even in the regime where
bandwidth binds. Compressing document tokens with product quantization and fusing
the asymmetric distance computation into the kernel reduces bytes moved by
$15.9\times$ and is $3.0\times$ slower than the exact FP16 tensor-core path.
Counters exclude bandwidth, cache behavior and occupancy, and identify
$6.5\times10^{9}$ dependent shared-memory gathers as the limit, a count that is
independent of code width. A CUDA implementation with an explicitly staged table
recovers $4.7\times$ and remains slower than exact scoring. The comparison is
between a bandwidth-bound tensor-core matrix product and a lookup-bound SRAM
gather, and on current hardware the matrix product is faster by $3.0\times$ at a
$15.9\times$ disadvantage in bytes.

Together these results name the two quantities a fused multi-vector scoring kernel
should be designed and evaluated against. The first is where the shape sits relative
to the roofline ridge. The second is how many dependent gathers the scoring path
performs. Our head-to-head comparison over 28 shapes supports the same reading
from a different direction. In the memory-bound regime, three independently
developed kernels land within a few percent of each other and within 6\% of the
hardware limit. That regime is close to exhausted. Outside it they differ by up to
$2.9\times$, according to how each one handles the query dimension.

\bibliographystyle{ACM-Reference-Format}

\end{document}